\documentclass[11pt]{article}

\usepackage{palatino}
\usepackage{mathpazo}
\usepackage{graphicx}
\usepackage[margin=1in]{geometry}
\usepackage{amssymb, amsthm, amsmath}

\newtheorem{theorem}{Theorem}
\newtheorem{lemma}[theorem]{Lemma}

\def\<{\langle}
\def\>{\rangle}
\def\N{\mathbb{N}}

\def\Q{\mathbb{Q}}

\def\R{\mathbb{R}}
\def\eps{\varepsilon}
\newcommand{\comment}[1]{}
\newcommand{\fid}{\operatorname{F}}

\newcommand{\ket}[1]{\ensuremath{\left|#1\right\rangle}}

\title{\bf\LARGE Characteristics of Universal Embezzling Families}
\author{
  Debbie Leung\footnote{%
    Institute for Quantum Computing and
    Department of Combinatorics and Optimization,
    University of Waterloo,
    Waterloo, Ontario, Canada.}
  \and
  Bingjie Wang\footnote{
  	University of Cambridge,
  	Cambridge, Cambridgeshire, United Kingdom.
  }
}
\date{6 July, 2014}  

\begin{document}
\maketitle

\begin{abstract}
We derive properties of general universal embezzling families for
bipartite embezzlement protocols, where any pure state can be
converted to any other without communication, but in the presence of
the embezzling family. Using this framework, we exhibit various
families inequivalent to that proposed by van Dam and Hayden.
We suggest a possible improvement and present detail numerical 
analysis.
\end{abstract}

\section{Introduction}
\label{primer}

We begin by defining bipartite {\it quantum state embezzlement}
between Alice and Bob.  Let $\ket{\varphi}$ and $\ket{\mu}$ be
bipartite quantum states; embezzlement of $\ket{\varphi}$ from
$\ket{\mu}$ is the transformation $\ket{\mu} \mapsto
\ket{\mu}\ket{\varphi}$ using only local operations. Operationally, 
Alice and Bob
share $\ket{\mu}$ and, without further communication, ``embezzle'' a
shared $\ket{\varphi}$.  

Pure bipartite entangled states, their interconversions, and their
applications in quantum information processing tasks have been
well-studied.  Axiomatically, entanglement, as a quantum correlation,
does not increase without communication, rendering exact embezzlement
impossible for a general $\ket{\varphi}$.  Surprisingly, van Dam and
Hayden \cite{DH02} showed embezzlement can be approximated, with
arbitrary precision, as the dimension of $\ket{\mu}$ grows.
Furthermore, arbitrary $\ket{\varphi}$ can be embezzled from the same
$\ket{\mu}$.  We call such a sequence of states \ket{\mu(n)} a {\it
  universal embezzling family}.


Embezzlement has found interesting applications.  It 
enables remote parties to share an arbitrary state on
demand without communication (see for example \cite{V-private}).
Furthermore, embezzlement hides the existence or the disappearance of
a quantum state from any external observer.  Thus embezzlement is used
in the noisy channel simulation in the original \cite{QRSTa} and an
alternative \cite{QRSTb} proof of the quantum reverse Shannon theorem.
Finally, in \cite{LTW08}, embezzlement motivates a game for which no
finite amount of entanglement suffices in an optimal strategy, and
provides proofs that some natural classes of quantum operations are
not topologically closed.

The results in \cite{DH02} have been extended in several ways.  An
alternative embezzling family for any number of parties is proposed in
\cite{LTW08}.  This family also achieves better approximation for a
given dimension of $|\mu\>$ for non-universal embezzlement (a method
attributed to \cite{HS-private}). Reference \cite{V-private} provides
an embezzlement protocol that is robust against discrepancy
between the descriptions of $\ket{\varphi}$ available to Alice and Bob. 


There are many unresolved questions concerning embezzlement.  In the
multiparty setting, the only known universal multiparty embezzlement
family is an $\epsilon$-net of the non-universal embezzlement states
\cite{LTW08}; perhaps more efficient universal families exist.  In the
bipartite setting, the family in \cite{DH02} is not known to be
optimal, but it has been elusive to find an optimality proof or a
better family.  Likewise, there may be a lower dimensional resource
state for the robust protocol in \cite{V-private}.  Very few universal
embezzling families are known, and finite size effect or the
computational complexity of embezzlement is hardly studied.

In this paper, we focus on the bipartite setting.  We derive
conditions for universal embezzlement, and exhibit a countably
infinite number of inequivalent families.  We conjecture a universal
embezzling family based on our findings, and provide numerical
evidence for the improvement in efficiency.

During the preparation of this manuscript, we learned of the result by
Dinur, Steurer, and Vidick reported in \cite{V-private}, and another
on-going study of embezzlement by Haagerup, Scholz, and Werner
\cite{HSW-private}.

\subsection*{Canonical Form for Embezzlement}
Any pure bipartite state has a Schmidt decomposition (see for example
\cite{NC00}).  Since the parties can perform local unitary operations, 
without loss
of generality, $\ket{\mu} = \sum_{i=1}^{ \tilde{n}} \mu_i |i\>_{A_1}
|i\>_{B_1}$ and $\ket{\varphi} = \sum_{j=1}^{m} \varphi_j |j\>_{A_2}
|j\>_{B_2}$ where $\{\ket{i}\}_{i=1}^{\tilde{n}}$,
$\{|j\>\}_{j=1}^{m}$ are orthonormal bases for Alice's systems $A_1,
A_2$, and for Bob's systems $B_1, B_2$.  Furthermore, $\mu_i$,
$\varphi_i$ can be chosen non-negative and decreasing, with $\sum
\mu_i^2 = 1$ and $\sum \varphi_i^2 = 1$ so \ket{\mu} and \ket{\varphi}
are normalized.  We refer to $\tilde{n}$ as the Schmidt rank and the
$\mu_i$s as Schmidt coefficients of \ket{\mu}; the same terminology
holds for the Schmidt decomposition of any bipartite state.

In this canonical form, there is an exchange symmetry between Alice
and Bob.  Furthermore, any quantum operation can be implemented as an
isometry, $U$, with possibly larger output space. In embezzlement, the
actual output state is $U \otimes U \ket{\mu}$.

\subsection*{Measure of success and optimal strategy}
One measure of the precision of the embezzlement protocol is the
fidelity.  The fidelity between two pure states 
\ket{\varphi} and \ket{\psi} is given by 
$F(\ket{\varphi},\ket{\psi}) = |\<\varphi|\psi\>|$ 
(see \cite{NC00}).  From
\cite{Vidal-Jonathan-Nielsen-00}, it follows that the fidelity between
the output $U \otimes U \ket{\mu}$ and the target $\ket{\mu}
\ket{\varphi}$ is optimized by the isometry $U$ taking $A_1 \mapsto A_1A_2$
(likewise for Bob) that simply permutes the basis states, such that
$\ket{\omega} {:}= U^\dagger \otimes U^\dagger \ket{\mu}
\ket{\varphi}$ has decreasing Schmidt coefficients.
The optimal fidelity is $\max_U \<\mu| \<\varphi| \, (U \otimes U
\ket{\mu}) = \<\omega| \, (|\mu\> \otimes |1\>|1\>)$.
The state $|\mu\> \otimes |1\>|1\>$ has 
Schmidt coefficients $\mu_i$ followed by zeros.  
We denote it by the equivalent state $|\mu\>$ throughout. 

In this paper, we only consider embezzlement protocols that involve
permutation of the basis states.  We often consider ``optimal
embezzlement'' as described above.  Given a universal embezzling
family, we focus on a subsequence \ket{\mu(n)} indexed by the local
dimension $n = \tilde{n}$.

Intuitively, a state $|\mu\>$ is useful for universal embezzlement if its
Schmidt coefficients $\mu_i's$ has high fidelity with respect to 
$\{\mu_i \varphi_j\}_{ij}$ for any valid $\{\varphi_j\}$.

\subsection*{General vs regular embezzling families}
The most general embezzling family has the form 
%
%
\[
\ket{\mu(n)} = \sum_{i=1}^n \mu(i, n)\ket{i}\ket{i}
\]
where for each $n$, $\mu(i, n)$ is decreasing with $i$ and 
$\sum_{i=1}^n \mu(i, n)^2 = 1$.  
An interesting special case concerns embezzling families whose Schmidt
coefficients are generated by decreasing functions of one variable
$i$, $f : \N \mapsto \R^+$.
They are given by 
\[
\ket{\mu(f,n)} = \frac{1}{\sqrt{C(f,n)}}\sum_{i=1}^n f(i)\ket{i}\ket{i} \,.
\]
where $C(f, n) = \sum_{i=1}^n f(i)^2$ so \ket{\mu(f,n)} is normalized. 
We call these universal embezzling families ``regular''.  
They are a direct generalization of 
the universal embezzling family proposed in \cite{DH02}:
\[
\ket{\mu(f_{dh}, n)} = \frac{1}{\sqrt{C(f_{dh},n)}}
			\sum_{i=1}^{n} \frac{1}{\sqrt{i}}\ket{i}\ket{i}
\]
where $f_{dh}(x) = 1/\sqrt{x}$.

\section{Properties of Embezzling Families}
\label{general}
In this section, we present necessary conditions and sufficient
conditions for a sequence, \ket{\mu(n)}, to be a universal embezzling
family.

First, for universal embezzlement, it suffices to be able to embezzle
any Schmidt rank 2 state. We first introduce a lemma stating that
embezzlement of different Schmidt rank $m$ states can be done in
superposition.  This result is a simple generalization of both  
embezzlement and coherent state exchange \cite{LTW08}.

\begin{lemma} \label{embinsuperposition}
Suppose it is possible to embezzle any \ket{\varphi} with Schmidt rank $m$ 
using $|\mu\>$ with fidelity at least $F$ (see Section \ref{primer}), 
then the following transformation
\[
\sum_{j=1}^{k} \alpha_j |\mu\> |jj\> \rightarrow 
\sum_{j=1}^k \alpha_j |\mu\> |\varphi_j\>
\]
can be performed with fidelity at least $F$ without communication, for 
any $\alpha_j$'s satifying $\sum_{j=1}^k |\alpha_j|^2 = 1$ and for each 
$|\varphi_j\>$ of the form 
\[
 |\varphi_j\> = \sum_{l=1}^m \varphi_{j,l} |m(j{-}1){+}l\>|m(j{-}1){+}l\> \, 
 ~\text{with}~ \sum_{l=1}^m |\varphi_{j,l}|^2 = 1 
\,. 
\]
\end{lemma}
\begin{proof} 
The given embezzlement property, as specified in Section \ref{primer}, 
implies that 
$\forall j, \exists U_j$ such that 
$F(U_j \otimes U_j |\mu\> |11\>, |\mu\> \sum_{l=1}^{m} \varphi_{j,l} |ll\>) 
\geq F$.  
Modifying the input and output bases gives a $\widetilde{U}_j$ such
that $F(\widetilde{U}_j \otimes \widetilde{U}_j |\mu\> |jj\>, |\mu\>
|\varphi_j\>) \geq F$.  Further define $\widetilde{U}_j |\xi\> |j'\> =
0$ for all $|\xi\>$ whenever $j' \neq j$.  So, $U = \sum_j
\widetilde{U}_j$ is an isometry satisfying: 
\[
\left[ \sum_{j'=1}^k \alpha_{j'}^* \<\mu| \<\varphi_{j'}| \right] 
        \left[ U \otimes U \sum_{j=1}^{k} \alpha_j |\mu\> |jj\> \right]
= 
\left[ \sum_{j'=1}^k \alpha_{j'}^* \<\mu| \<\varphi_{j'}| \right] 
        \left[ \sum_{j=1}^k \alpha_j \widetilde{U}_j \otimes 
          \widetilde{U}_j |\mu\> |jj\> \right]
\geq F  \,.
\]
\end{proof}

We now analyze embezzlement of general states by recursively
embezzling Schmidt rank $2$ states while reusing the embezzlement
state.  To do so, we use two facts concerning the {\it trace distance}
between two density matrices $\sigma_{1,2}$ of equal dimension,
defined as $T(\sigma_1, \sigma_2) {:}{=} \, \frac{1}{2} \| \sigma_1 -
\sigma_2 \|_1$ where $\|\cdot\|_1$ denotes the Schatten $1$-norm.  
First, for two pure states,
$T(\ket{\sigma_1},\ket{\sigma_2})^2 + 
F(\ket{\sigma_1}, \ket{\sigma_2})^2 = 1$.  Second, the trace
distance is nonincreasing under any quantum operation and is subadditive. 
(See \cite{Ruskai94,FuchsvdG99,NC00} for detail.)   
In particular, if $F(|\sigma\>, U|\sigma_1\>) \geq F_1$ and 
$F(|\sigma_1\>,|\sigma_2\>) \geq F_2$, then, 
\begin{equation}
\label{subaddtrdist}
\sqrt{1-F(|\sigma\>,U|\sigma_2\>)^2} = T(|\sigma\>, U|\sigma_2\>) \leq
T(|\sigma\>, U|\sigma_1\>) + T(|\sigma_1\>, |\sigma_2\>)
\leq \sqrt{1-F_1^2} + \sqrt{1-F_2^2} \,, 
\end{equation}
which bounds the performance of substituting $|\sigma_1\>$ by $|\sigma_2\>$
in any operation $U$.  

\begin{lemma} \label{only2qubits}
Suppose it is possible to embezzle any Schmidt rank $2$ state from
$|\mu\>$ with fidelity at least $F$. Then,
embezzlement of any Schmidt rank $m$ state $|\varphi\>$ can be achieved 
with fidelity at least $F_m$ where $1-F_m^2 \leq \lceil \log_2 m \rceil^2 
(1-F^2)$.
\end{lemma}
\begin{proof} 

It suffices to prove the theorem for $m=2^l$ for $l \in \N$ via
induction on $l$. The base case $l=1$ is given. Assume, for some $k$,
for any state \ket{\phi} with Schmidt rank at most $2^k$, there exists
an isometry $V$, such that $F_k = F(V \otimes
V\ket{\mu},\ket{\mu}\ket{\phi})$ satisfies $1 - F_k^2 \leq k^2 (1 -
F^2)$.

It remains to show that any $\ket{\varphi} = \sum_{i=1}^m \varphi_i 
\ket{i}\ket{i}$ with $m = 2^{k+1}$ can be embezzled with the desired 
fidelity.
To do so, let $\alpha_j^2 = \varphi_{2j-1}^2 + \varphi_{2j}^2$
and $|\varphi_j\> = \alpha_j^{-1} (\varphi_{2j{-}1}|2j{-}1\>|2j{-}1\> + 
\varphi_{2j}|2j\>|2j\>)$ for $j=1, 2, 3, \cdots, 2^k$. Apply the induction 
hypothesis; so $\ket{\phi} = \sum_{j=1}^{m/2} \alpha_j \ket{jj}$ 
can be embezzled 
with fidelity at least $F_k$ with some isometry $V$. In addition, from Lemma
\ref{embinsuperposition}, $|\mu\>\ket{\phi} \rightarrow |\mu\> 
(\sum_{j=1}^{m/2} \alpha_j |\varphi_j\>) = |\mu\>|\varphi\>$ can be performed 
with fidelity at least $F$. Finally, using Eq.~(\ref{subaddtrdist}), 
we evaluate the 
fidelity of composing these two steps by taking $|\sigma\> = |\mu\>|\varphi\>$,
$|\sigma_1\> = |\mu\> \ket{\phi}$, and $\ket{\sigma_2} = V \otimes V |\mu\>$.
This yields $\sqrt{1-F_{k+1}^2} \leq 
\sqrt{1-F^2} + \sqrt{1-F_k^2} \leq  
\sqrt{1-F^2} + \sqrt{k^2(1 - F^2)}   
=(k + 1)\sqrt{1 - F^2}.$
\end{proof}

{\it Remark.}  Due to Lemma \ref{only2qubits}, we take $\ket{\varphi}
= \alpha\ket{00} + \beta\ket{11}$ unless otherwise stated. In
\ket{\omega}, the Schmidt coefficients either have the form $\alpha
\mu(i, n)$ or $\beta \mu(i, n)$ which we will refer to as $\alpha$ and
$\beta$ terms respectively.

Our next observation implies the divergence of the normalization
factor $C(f,n)$ for regular embezzling families.

\begin{lemma}\label{divergentcn}
If \ket{\mu(n)} is a universal embezzling family, then 
$\mu(1,n) \rightarrow 0$ as $n \rightarrow \infty$.  In particular, 
for regular universal embezzling families, $C(f, n)\rightarrow 
\infty$ as $n \rightarrow \infty$. 
\end{lemma}
\begin{proof} 
Let $F$ be the fidelity of the embezzlement protocol, minimized over
$|\varphi\>$. Lower bound $1-F$ by considering specifically 
$|\varphi\> = (|11\>+|22\>) / \sqrt{2}$:
\[
1 - \fid(\ket{\mu(n)}, \ket{\omega}) = 
1 - \sum_i \mu_i \omega_i = 
\frac{1}{2} \sum_{i=1}^{2n} (\mu_i - \omega_i)^2 
\geq \frac{1}{2} (\mu_1 - \omega_1)^2 = 
\frac{1}{2} \left(1 - \frac{1}{\sqrt{2}}\right)^2 \mu_1^2 \,,
\]
where we use the shorthard $\mu_i$ for $\mu(i,n)$, $\omega_i$'s are 
the Schmidt
coefficients of $\ket{\omega}$ in decreasing order, 
and $\omega_1 = \mu_1/\sqrt{2}$.
Since $\mu_1 > 0$ (else \ket{\mu(n)} cannot be a valid quantum state),   
$F \rightarrow 1$ implies $\mu(1,n) \rightarrow 0$ as $n \rightarrow \infty$. 

For regular families \ket{\mu(f,n)}, $\mu(1,n) = f(1)/\sqrt{C(f,n)}$, so 
$C(f, n)\rightarrow \infty$ as $n \rightarrow \infty$.
\end{proof}

Note that \ket{\mu(f, n)} $=$ \ket{\mu(cf, n)} for any constant
$c$. Thus, we consider the {\it order} of a regular universal
embezzling family defined as follows: a universal embezzling family
has {\it order} $g$ if and only if $C(f,n)=\Theta(g)$, {\it e.g.},
\ket{\mu(f_{dh},n)} has order $\ln n$.  Lemma \ref{divergentcn} shows
that the ``misalignment'' of the first terms of \ket{\mu(n)} and
\ket{\omega(n)} has to be corrected by a divergent order.  

The next lemma gives a sufficient condition in terms of the asymptotic
behavior of the ratio $\rho(\ket{\varphi}, f, i) \, {:}{=} \,
\omega_i/\mu_i$.  First, given $\ket{\varphi} = \sum_{j=1}^m \varphi_j
|j\>|j\>$ and $f$, we explain how to make this ratio well-defined for
all $i \in \N$.
Fix an arbitrary $n$ and let $\mu(i,n) = f(i)/\sqrt{C(f,n)}$ 
for $i=1,\cdots,n$.
Let $\omega(i,n)$ be the $i$-th largest element in 
$S_n = \{ \mu(i,n) \varphi_j \}$.
Define $\rho(\ket{\varphi}, f, i)$ to be $\omega(i,n)/\mu(i,n)$ for
$i=1,\cdots,n$.  Note that the $\sqrt{C(f,n)}$ factors cancel out in
the ratios.
Furthermore, let $n'>n$ and define $\omega(i,n')/\mu(i,n')$ for $i =
1,\cdots, n'$ similarly.  The first $n$ ratios coincide with
$\omega(i,n)/\mu(i,n)$ because the $n$ largest terms in $S_{n'} = \{
\mu(i,n') \varphi_j \}$ are labeled by the same $(i,j)$'s as those in
$S_n$.

\begin{lemma} \label{convergenttl}
Let $f : \N \mapsto \R^+$ be a decreasing function with 
$C(f, n)\rightarrow\infty$.  
If $\forall\ket{\varphi}$, $\rho(\ket{\varphi},f,i) 
\rightarrow 1$, then $\ket{\mu(f, n)}$ forms a 
regular universal embezzling family.
\end{lemma}
\begin{proof}
Since $\rho \rightarrow 1$, given any $\eps > 0$, $\exists n_\eps$ such that 
$(1-\eps)\mu_i < \omega_i < (1+\eps)\mu_i$ for all $i > n_\eps$.  Thus 
\begin{align*}
\fid(\ket{\mu(f, n)}, \ket{\omega}) & = \sum_{i=1}^n\mu_i\omega_i \; = \; 
\sum_{i=1}^{n_\eps}\mu_i\omega_i \; + \! \! \sum_{i={n_\eps}+1}^n\mu_i\omega_i
\; \; > \sum_{i={n_\eps}+1}^n\mu_i\omega_i 
\\
& > (1-\eps)\sum_{i={n_\eps}+1}^n \mu_i^2 
\; > \; (1-\eps) - \sum_{i=1}^{n_\eps} \mu_i^2 
\; > \; (1-\eps) - \frac{C(f,n_\eps)}{C(f,n)} \,.
\end{align*}
Since $n_\eps$ does not depend on $n$, and $C(f, n) \rightarrow
\infty$, $\fid(\ket{\mu(f, n)}, \ket{\omega}) \rightarrow 1$. Thus,
\ket{\mu(f, n)} forms a universal embezzling family. In fact, 
$1-F < \eps + C(f,n_\eps) / C(f,n)$.  
\end{proof}

We note on the side that Lemma \ref{convergenttl} does not have a
natural converse.  Universal embezzling families may exist with
infinitely many but intermittent violations of the condition
$\rho(\ket{\varphi}, f, i) \approx 1$.

\section{Variations on \ket{\mu(f_{dh}, n)}}

In this section and the next, we focus on regular universal embezzling
families.  We consider the ``simplest'' variation from $f_{dh}$, which
is $f = g / \sqrt{x}$. This construction can be used in two ways to
yield a universal embezzling family.

\begin{lemma} \label{h_construction}
Let $h: \N \rightarrow \R^+$. If, $C(f, n) \rightarrow \infty$, 
$f = h / \sqrt{x}$ is decreasing, and 
$h(kx + c) / h(x) \rightarrow 1$ as $x \rightarrow \infty$ 
for any constant $k \in \N$, $c \in \N \bigcup \{0\}$, 
then, \ket{\mu(f, n)} forms a universal embezzling family.
\end{lemma}

\begin{proof}
First, if 
$h(kx + c) / h(x) \rightarrow 1$ as $x \rightarrow \infty$,
for any constants $k \in \N$, $c \in \N \bigcup \{0\}$,
then, 
$h(k_1x + c_1) / h(k_2x + c_2) \rightarrow 1$ as $x \rightarrow \infty$ 
for any constants $k_1, k_2 \in \N$ and $c_1, c_2 \in \N \bigcup \{0\}$.  
This follows from the quotient rule 
\[
   \lim_{x \rightarrow \infty} 
   \frac{h(k_1x + c_1)}{h(k_2x + c_2)} = 
   \frac{ \lim_{x \rightarrow \infty} \frac{h(k_1x + c_1)}{h(x)} }
        { \lim_{x \rightarrow \infty} \frac{h(k_2x + c_2)}{h(x)} }
   = 1 \,.
\]

Following Lemma \ref{only2qubits}, 
consider $\ket{\varphi} = \alpha\ket{11} + \beta
\ket{22}$. Let $z = (\alpha / \beta)^2$. Recall that the optimal fidelity is 
achieved with decreasing Schmidt coefficients $\omega_i$ for $|\omega\>$. 
Here, we consider a particular ordering of Schmidt coefficients, 
$\ket{\widetilde{\omega}}$, which can be suboptimal. 
Then, any lower bound on $\fid(\ket{\mu(f,n)}, 
\ket{\widetilde{\omega}})$ also applies to $\fid(\ket{\mu(f,n)}, \ket{\omega})$.

First, suppose $z = p/q \in \Q$. Call the $p$ largest $\alpha$-terms (see 
remark to Lemma \ref{only2qubits}) the first $\alpha$-block, the next $p$
largest $\alpha$-terms the second $\alpha$-block, and so on. 
Define the $\beta$-blocks
similarly, but with block size $q$ instead. Construct 
$\ket{\widetilde{\omega}}$ such that the $l$-th block of $p+q$ terms comes
from the $l$-th $\alpha$- and $\beta$-blocks.  In other words, for 
$l(p+q)+1 \leq i \leq (l+1)(p+q)$:
\[
  \widetilde{\omega}_i = \begin{cases} \alpha f(lp + C_1) / C(f, n) 
                          \text{ or } \\
                          \beta f(lq + C_2) / C(f, n) \end{cases}
\]
where $1 \leq C_1 \leq p$ and $1 \leq C_2 \leq q$. Now consider $\widetilde{
\omega}_i / \mu_i$ where $i = l(p+q) + C$ for any $0 \leq C \leq p+q$. If 
$\widetilde{\omega}_i$ is an $\alpha$-term, then 
\[
  \frac{\widetilde{\omega}_i}{\mu_i} = \alpha \sqrt{\frac{l(p+q) + C}{lp + C_1}}
  \cdot\frac{h(lp + C_1)}{h( \,l (p{+}q) + C)} \,.
\]
As $i \rightarrow \infty$, $l \rightarrow \infty$, 
$h(lp + C_1) / h(l(\,p{+}q)+C) \rightarrow 1$, 
so $\widetilde{\omega}_i / \mu_i \rightarrow \alpha \sqrt{(p+q)/p} = 1$. If 
$\widetilde {\omega}_i$ is a $\beta$-term, with a similar argument, 
$\widetilde\omega_i / \mu_i \rightarrow \beta \sqrt{(p+q)/q} = 1$. 
Then, by Lemma \ref{convergenttl}, 
$\fid(\ket{\mu(f,n)}, \ket{\widetilde{\omega}})\rightarrow1$. 

If $z \not\in\Q$, the above proof applied to rational approximations of 
$z$ provides the desired result. More specifically,  
if $z = (\alpha / \beta)^2 \not\in\Q$, $\forall\delta > 0$, 
$\exists z^\prime = p/q \in \Q$ such that
\begin{equation} \label{bound}
\left(\frac{\alpha}{\beta}\right)^2 - \delta < \frac{p}{q} < 
  \left(\frac{\alpha}{\beta}\right)^2 + \delta \,.
\end{equation}
The previous argument shows that $\widetilde{\omega}_i/\mu_i$ tends to either 
$\alpha\sqrt{(p+q) / p}$ or $\beta\sqrt{(p+q) / q}$. Eliminating $p/q$ 
in these expression using (\ref{bound}) gives:
\[
1 - \frac{\delta\beta^4}{\alpha^2 + \delta\beta^2} < \alpha\sqrt{\frac{p+q}{p}}
    < 1 + \frac{\delta\beta^4}{\alpha^2 + \delta\beta^2} ~~\text{ and }~~
1 - \delta\beta^2 < \beta\sqrt{\frac{p+q}{q}} < 1 + \delta\beta^2 
\]
and both quantities tend to $1$ as $\delta \rightarrow 0$.
\end{proof}

\begin{lemma} \label{g_construction}
Let $g: \N \mapsto \R^+$ be an increasing function such that $f = g / 
\sqrt{x}$ is decreasing. If, in addition, $\forall m \in \N, C(f, n/m) / C(f, n)
\rightarrow 1$ as $n \rightarrow \infty$, then \ket{\mu(f,n)} forms a universal
embezzling family.
\end{lemma}
\begin{proof}
This proof derives heavily from \cite{DH02}.  

Claim: $\forall j \,, ~\omega_j \leq \mu_j$. Let $N(t) = 
|\{l: \mu_t < \omega_l\}|$. The claim is equivalent to $N(t) < t$ as
$\{\omega_l\}$ is decreasing. 
Since $\omega_l = \varphi_i f(j) / C(f,n)$ for some $i,j$, we let
$N_i^t = | \{j: \mu_t < \varphi_i f(j) / C(f,n)\}|$. Now, 
\[ 
   \mu_t < \varphi_i \frac{f(j)}{C(f,n)} ~~\Leftrightarrow~~ 
   f(t) < \varphi_i f(j) ~~\Leftrightarrow~~ 
   \frac{jg(t)^2}{tg(j)^2} < \varphi_i^2 \,.
\]
We can infer that $t \leq j$ since 
the middle inequality implies $f(j) < f(t)$ and $f$ is decreasing. 
Then, the last inequality and the monotonicity of $g$ imply that  
$j < \varphi_i^2 t$, so $N_i^t < \varphi_i^2 t$ and $N(t) = \sum_i N_i^t < t$
(recall the normalization $\sum_i \varphi_i^2 = 1$). Finally, 
\begin{equation} \label{g_fid}
\fid(|\mu(f,n)\>,\ket{\omega}) = \sum_{i=1}^n \mu_i \omega_i 
\geq \sum_{i=1}^n \omega_i^2 
\geq \sum_{j=1}^{\lfloor n/m \rfloor}\sum_{i=1}^m
	\frac{\varphi_i^2 f(j)^2}{C(f,n)} = 
        \frac{C(f, \lfloor n/m \rfloor )}{C(f, n)} \rightarrow 1
\end{equation}
where the last inequality comes from replacing the sum with possibly 
fewer and smaller terms.
\end{proof}

Lemma \ref{g_construction} states that $f$ can fall off slower than $f
= 1 / \sqrt{x}$ as long as $C(f, n/m) / C(f, n) \rightarrow 1$.

\section{New classes of regular universal embezzling families} 


Now we present two sequences of regular universal embezzling families using
Lemmas \ref{h_construction} and \ref{g_construction}. First, define 
$\lambda(x) = \ln (x + e)$ and its n-fold composition:
$\lambda^0(x)=x$, $\lambda^1(x)=\ln(x+e)$, $\lambda^2(x)=\ln(\ln(x+e) + e)$, 
and so on.

Now define the $G$ and $H$ functions of class $r$ as:
\begin{equation}	
	G_r(x) = \frac{1}{\sqrt x} \prod_{s=1}^r \sqrt{\lambda^s(x)}
\end{equation} \begin{equation}
	H_r(x) = \frac{1}{\sqrt x} \prod_{s=1}^r \frac{1}{\sqrt{\lambda^s(x)}}
\end{equation}

For every $r$, we will see that \ket{\mu(G_r, n)} and \ket{\mu(H_r, n)} have
different orders and are universal embezzling families. Therefore, the number
of orders for regular universal embezzling families is infinite. 

To estimate $C(H_r, n)$, we use integral approximations: 
\[
\frac{d}{dx} \lambda^{r+1}(x) = 
  \prod_{s=0}^{r} {1 \over \lambda^s(x) + e} \approx
  H_r(x)^2 
\Rightarrow \sum_{i=1}^n H_r(i)^2 \approx 
  \int_{1}^n H_r(x)^2dx \approx
  \lambda^{r+1}(n)
\]
Thus, the order of \ket{\mu(H_r, n)} is $\lambda^{r{+}1}(n) \approx 
\ln^{r{+}1}(n)$ for large $n$.

For $C(G_1, n)$, we apply integral approximations and the inequality 
$G_1(x)^2 \leq (\ln(x{-}e))/(x{-}e)$ for $x \geq 5$ to obtain:
\begin{equation}
\label{g_estimate}
  \int_1^n \frac{\ln(x+e)}{x+e} < \int_1^n \frac{\ln(x+e)}{x} \approx 
  \sum_{i=1}^n G_1(i)^2 \leq \sum_{i=1}^5 G_1(i)^2 
  + \int_5^n \frac{\ln(x-e)}{x-e} \,.
\end{equation}
The integrals are all well approximated by $(\ln n)^2/2$.  Thus 
$C(G_1,n) = \Theta[(\ln n)^2]$.  
For general $C(G_r, n)$, there is no simple approximation, but we can
show that subsequent orders are progressively ``higher.''  
First, 
\begin{equation}
   C(G_{r{+}1},n) = \sum_{i=1}^n G_{r{+}1}(i)^2 
   = \sum_{i=1}^n \lambda^{r{+}1}(i) \, G_{r}(i)^2 
   \geq \sum_{i=1}^n G_{r}(i)^2 = C(G_{r},n) \,.
\label{eq:Gorder}
\end{equation}
We show by contradiction that $C(G_{r+1},n) \not = \Theta[C(G_r, n)]$. 
If so, there are constants $\kappa,n_0$, such that 
$\forall n > n_0$, $C(G_{r+1}, n) \leq \kappa C(G_r, n)$. 
Pick $n_1 > n_0$ so that $\lambda^{r{+}1}(n_1) \geq 3 \kappa$, 
and $n_2 > n_1$ such that 
$\sum_{i=1}^{n_1} G_r(i)^2 \leq \sum_{i=n_1{+}1}^{n_2} G_r(i)^2$.  Now,  
\[
 \kappa C(G_r,n_2) 
\leq  
 2 \kappa \sum_{i=n_1{+}1}^{n_2} G_r(i)^2
\leq
 \frac{2}{3} \lambda^{r{+}1}(n_1) \sum_{i=n_1{+}1}^{n_2} G_r(i)^2 
\leq  
 \frac{2}{3} \sum_{i=n_1{+}1}^{n_2} G_r(i)^2 \lambda^{r{+}1}(i)^2
\leq 
 \frac{2}{3} C(G_{r{+}1},n_2)
\]
a contradiction.


\subsection*{Embezzling Properties of $|\mu(G_r,n)\>$}

First, we sketch that $G_r(x)$ is decreasing. Let $t(x)=\lambda(x)/\sqrt{x}$.
Then, $t(x)$ is decreasing because its first derivative has 
the same sign as $\theta(x) = 2x - (x+e) \ln (x+e)$, and $\forall x>0$,  
$\theta(x) < 0$ because its first derivative is negative and 
$\theta(0) < 0$.
Therefore,  $\lambda(x+1) / \sqrt{x + 1} <
\lambda(x) / \sqrt{x}$ and $\lambda(x+1)/\lambda(x) < \sqrt{x+1} / \sqrt{x}$ 
for $x > 0$. Repeating this result yields: $\lambda^2(x+1)/\lambda^2(x) < 
\sqrt{\lambda(x+1)}/ \sqrt{\lambda(x)} < [(x + 1) / x]^{1/4}$, etc. Now:
\[
\frac{G_r(x+1)^2}{G_r(x)^2} = 
	\frac{x}{x+1}\prod_{i=1}^r\frac{\lambda^i(x+1)}{\lambda^i(x)} <
	\frac{x}{x+1}\prod_{i=1}^r\left[\frac{x+1}{x}\right]^{1/2^i} < 1
\]
so the positive functions $G_r$ are all decreasing. 

Second, $\forall r \geq 1, C(G_r, n)$ diverges (see Eq.~(\ref{eq:Gorder})).  

We can establish that \ket{\mu(G_1, n)} forms a universal embezzling family
using Lemma \ref{g_construction}, by using the estimate (\ref{g_estimate}) 
to conclude that 
\[ 
\frac{C(G_1, n/m)}{C(G_1, n)} \sim \left(1 - \frac{\ln m}{\ln n}\right)^2 \,. 
\]
However, the lower 
bound for fidelity of embezzlement by \ket{\mu(G_1, n)} is no better than 
that of \ket{\mu(f_{dh}, n)}, despite Lemma \ref{convergenttl} (recall: 
$1 - F < \eps + C(f, N_\eps) / C(f, n)$) and the higher order of $G_1$. 

For other $G_r$, we will show that Lemma \ref{h_construction} applies. We 
first show by induction that $\forall s \in \N$, $\lambda^s(k x + c) / 
\lambda^s(x) \rightarrow 1$ for any constants $k, c$ 

For $s=1$:
\begin{equation}
\frac{\lambda^1(kx + c)}{\lambda^1(x)} 
= \frac{\ln(x + c/k) + \ln k}
       {\ln(x)} \rightarrow 1 \,.
\label{requals1}
\end{equation}

For $s \geq 2$, both $\lambda^{s{-}1}(kx + c) \rightarrow \infty$ and 
$\lambda^{s{-}1}(x) \rightarrow \infty$. By induction hypothesis,
their ratio tends to $1$. Thus, the proven base case (\ref{requals1}) implies 
$\lambda^s(kx + c) / \lambda^s(x) = \lambda^1( \lambda^{s{-}1}
(kx + c)) / \lambda^1( \lambda^{s{-}1}(x)) \rightarrow 1$.  

Then, for any class $r$, by the limit rule for products and the continuity of 
$\sqrt{\cdot}$ and $1/\cdot$ over the range of interest, both 
$\prod_{s=1}^{r}\sqrt{\lambda^s(x)}$ and $\prod_{s=1}^r1/\sqrt{\lambda^s(x)}$   
satisfy the condition in Lemma \ref{h_construction}.  Thus, all $3$ conditions 
in Lemma \ref{h_construction} holds for $G_r$ and $|\mu(G_r,n)\>$ forms a
universal embezzling family.

\comment{
We claim $[\frac{d}{dx}\lambda^r(k_1x + c_1)]/[\frac{d}{dx}\lambda^r(k_2x + c_2)] 
\rightarrow 1$ for any constants $k_1, k_2,c_1, c_2, \forall r \in \N$. Then 
it is sufficient to apply L'H\^{o}pital's and other limit rules to show 
$h(k_1x+c_1) / h(k_2x + c_2) \rightarrow 1$. Proof of claim is via induction 
over $r$: $r = 1$ is easy to check, if the hypothesis holds for $r = k$, then 
by the chain rule:
\[
	\lim_{x\rightarrow\infty}\frac{\frac{d}{dx} \lambda^{k+1}(k_1x + c_1)}
							{\frac{d}{dx} \lambda^{k+1}(k_2x + c_2)} 
	= \lim_{x\rightarrow\infty}\frac{\lambda^k(k_2x+c_2)+e}
									{\lambda^k(k_1x+c_1)+e}
	\cdot\lim_{x\rightarrow\infty}\frac{\frac{d}{dx}\lambda^{k}(k_1x + c_1)}
								 	   {\frac{d}{dx}\lambda^{k}(k_2x + c_2)}
\]
The first limit is $1$ by L'H\^{o}pital's and the induction step; the second is 
$1$ by the induction step.
}

\subsection*{Embezzling Properties of $|\mu(H_r,n)\>$}
$H_r$ is obviously decreasing $\forall r$. We have already shown 
that $\prod_{s=1}^r
1/\sqrt{\lambda^s(x)}$ satisfies the condition in Lemma \ref{h_construction}
and $C(H_r, n) \rightarrow \infty$ from our estimate of $C(H_r, n)$. 
Therefore, by Lemma \ref{h_construction}, \ket{\mu(H_r, n)} forms a universal
embezzling family.

However, \ket{\mu(f_{dh}, n)} performs better when embezzling 
any entangled state.  This follows from
Lemma \ref{divergentcn} and the fact $C(H_r, n) / C(f_{dh},n) \rightarrow 0$ 
as $n \rightarrow \infty$.

\subsection*{Entanglement of \ket{\mu(f_{dh}, n)}, 
             \ket{\mu(G_1, n)}, and \ket{\mu(H_1, n)}}
Another metric of embezzlement efficiency is the amount of entanglement
required in creating \ket{\mu}. For the original embezzling family proposed
in \cite{DH02}, using integral approximations:
\[
  \text{Ent}(\ket{\mu(f_{dh}, n)}) = -\sum_{i=1}^n \mu_i^2 \log_2(\mu_i^2)
  \approx -\sum_{i=1}^n \frac{1}{i}\frac{1}{\ln n} \log_2 
                                   \frac{1}{i}\frac{1}{\ln n} \,.
\]
Simplifying the above and using integral approximations, 
the leading term of $\text{Ent}(\ket{\mu(f_{dh}, n)})$ is $(\log_2 n)/2$.  

Similarly, we can estimate $\text{Ent}(\ket{\mu(G_1, x)})$. 
We use $C(G_1,n) \approx (\ln n)^2/2$  
and the approximation $\lambda(x) \approx \ln x$ to conclude 
that 
\[
  \text{Ent}(\ket{\mu(G_1, n)})   
  \approx -\sum_{i=1}^n \frac{\ln i}{i}\frac{2}{(\ln n)^2} \log_2 
                                   \frac{\ln i}{i}\frac{2}{(\ln n)^2}
  \approx \frac{2}{3} \log_2 n
\]
where the last estimate concerns only the lead term and uses 
integral approximations.  

Finally, we use $C(H_1,n) \approx \lambda^2(n) \approx \ln \ln n$ to
estimate $\text{Ent}(\ket{\mu(H_1, n)})$ which is $\approx (\log_2 n) /
(\ln \ln n)$.

For a fixed Schmidt rank, $\text{Ent}(|\mu(G_1,n)\>)$ and
$\text{Ent}(|\mu(f_{dh},n)\>)$ are of the same order.  Meanwhile,
$\text{Ent}(|\mu(H_1,n)\>) \ll \text{Ent}(|\mu(f_{dh},n)\>)$.
However, if one fixes the precision, a higher Schmidt rank is 
needed to embezzle using $H_1$ than $f_{dh}$.

\section{Outperforming \ket{\mu(f_{dh}, n)}?!}

In the previous sections, we examine regular families that 
do not have order $\ln n$.  There are interesting 
sequences that are not regular.  
One such sequence is  
presented in \cite{LTW08} (due to \cite{HS-private}):
\begin{equation}
	\ket{\mu(n)} = \sqrt{\frac{2}{N+1}} \sum_{k=1}^N 
      \sin\left(\frac{k\pi}{N{+}1}\right)
      \ket{00}^{\otimes k} \ket{\varphi}^{\otimes N - k + 1}
\label{eq:aram-peter}
\end{equation}
where $n = 2^N$.  This sequence enables the embezzling of the specific
state $\ket{\varphi}$ with fidelity at least $1-\pi^2/2N^2 = 1-\pi^2/2
(\log_2 n)^2$, a marked improvement over the provable lower bound 
$1-O(1/\log_2 n)$ of the fidelity 
achieved by \ket{\mu(f_{dh}, n)}.  The sequence in (\ref{eq:aram-peter}) 
also saturates an
upper bound of the fidelity proved in \cite{DH02}.  However, if 
$\ket{\varphi} = (\ket{11} + \ket{22})
/ \sqrt{2}$ and Alice and Bob want to embezzle 
$\ket{\varphi'} = \alpha \ket{11} + \beta \ket{22}$, 
the fidelity $\rightarrow (\alpha+\beta)/\sqrt{2}$ as $N \rightarrow \infty$ 
which is bounded 
away from $1$ when $\ket{\varphi'} \neq \ket{\varphi}$.  

Instead, we propose the following. Let 
$gh$ be defined, for fixed $n$, and for $x \in \N, 1 \leq x \leq n$ as:
\[
  gh(x) = \begin{cases} H_1(1) \text{ when } x = 1 \\
    H_1(x) \text{ when } C(gh, x - 1) \geq \ln(x) \\
    G_1(x) \text{ when } C(gh, x - 1) < \ln(x) \,. \end{cases}
\]
Then define $GH(x)$ for $x \in \N, 1 \leq x \leq n$ as $gh(x)$ with
elements in decreasing order (the $n$ dependence is implicit here) and
designate $\ket{\mu(n)} = \sum_{i=1}^n GH(i)\ket{i}\ket{i}$. Due to
the limited dependence on $n$, we can still define $C(GH, n)$ as
before, and it differs from $\ln n$ by at most $G_1(n)^2$ or
$H_1(n)^2$, but both $G_1(x)$ and $H_1(x) \rightarrow 0$ as $x
\rightarrow \infty$. Therefore, $C(GH, n) \rightarrow \ln n$ as $n
\rightarrow \infty$.

The precise performance of $\ket{\mu(n)}$ as a universal embezzling 
family is hard to analyse.  So, we {\em numerically} evaluate the 
optimal fidelity (see Section \ref{primer}) of embezzling three sample 
states: 
$\ket{\varphi_{+}} = (2\ket{00} + \ket{11}) / \sqrt{5}$,  
$\ket{\varphi_{*}} = (\sqrt{\pi {-} 1}\ket{00} + \ket{11}) / \sqrt{\pi}$, 
and 
$\ket{\varphi_{\circ}} = (\ket{00} + \ket{11}) / \sqrt{2}$, 
using $\ket{\mu(n)}$ for $n = 2^N$, $N=3, \cdots, 33$.  
For comparison, we also perform numerical optimization for the fidelity of 
embezzlement using $\ket{\mu(f_{dh}, n)}$.   

Figure \ref{fig:num} summarizes the result.  

\begin{figure}[!ht]
\centering
\includegraphics[bb = 58 198 529 575,width=4.8in]{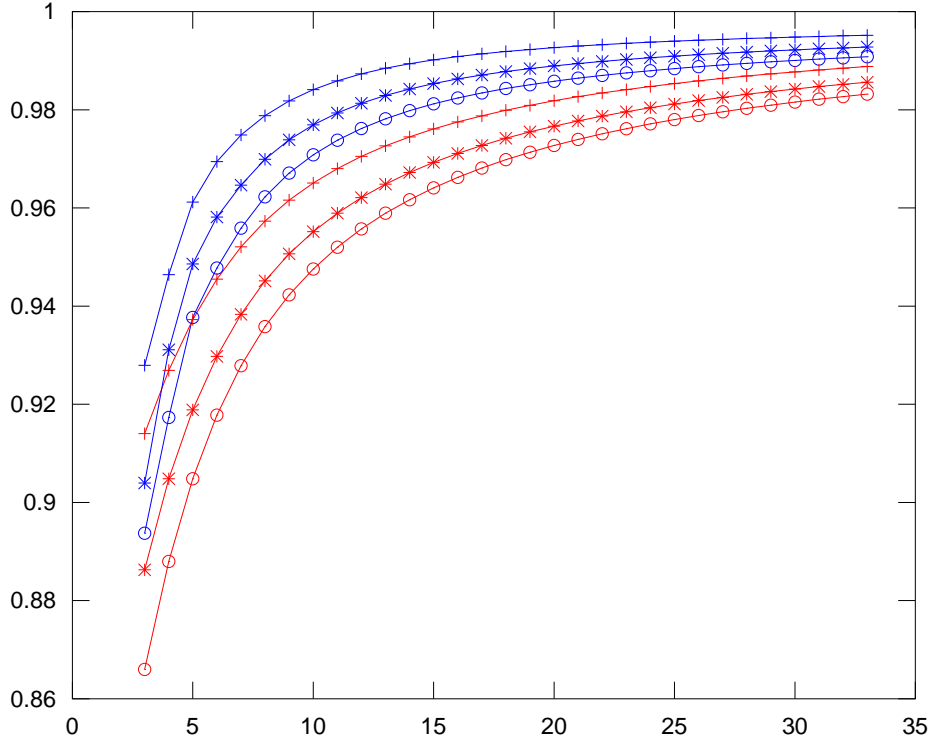}
\caption{Optimal fidelity of embezzlement as a function of the number
  of qubits ($N=\log_2 n$) held by each party.  The blue and red curves
  correspond to embezzlement using $\ket{\mu(GH, n)}$ and
  $\ket{\mu(f_{dh}, n)}$ respectively.  Data points marked by $+$, $*$,
  and $\circ$ correspond to $\ket{\varphi}$ being $\ket{\varphi_{+}}$, 
  $\ket{\varphi_{*}}$, and $\ket{\varphi_{\circ}}$ respectively.}
\label{fig:num}
\end{figure}


All calculations are done in IEEE double-precision.  The main source
of inaccuracy in the numerical optimization is the accumulation of
machine truncation errors in the calculation of $C(GH, n)$.  We
directly calculate $C(GH, n)$ for $3 \leq N \leq 26$ and approximate 
$C(GH, n)$ by $\ln n$ for $18 \leq N \leq 33$.  
The two methods yield optimal fidelities differing by 
less than $2 \times 10^{-6}$ for $18 \leq N \leq 26$.  


A quick inspection of Figure \ref{fig:num} suggests that
$\ket{\mu(n)}$ is indeed a universal embezzling family.  Furthermore,
$\ket{\mu(n)}$ outperforms $\ket{\mu(f_{dh}, n)}$ for the specific
cases studied.


We extrapolate the data to try to understand the asymptotic 
behavior of $\ket{\mu(n)}$.  
The least square fits to the optimal fidelities to embezzle 
$\ket{\varphi_{+}}, \ket{\varphi_{*}}, \ket{\varphi_{\circ}}$ 
using $\ket{\mu(n)}$ are:  
\begin{eqnarray}
F_{+} & = 0.9980 - 0.0759/N - 0.6358/N^2
\nonumber
\\
F_{*} & = 0.9976 - 0.1395/N - 0.6691/N^2
\nonumber
\\
F_{\circ} & = 0.9974 - 0.1971/N - 0.6862/N^2 \,.
\nonumber
\end{eqnarray}
The fitting parameters are insensitive to the method used to generate
$C(GH, n)$.  When fitting the data for $N_0 \leq N \leq 33$, the
fitting parameters are slightly sensitive to $N_0$.  We show the fits
for $N_0 = 10$, when the constant term is smallest, the magnitude for
the coefficients of the $1/N$ and $1/N^2$ terms are smallest and
largest respectively.  For $N_0$ ranging from $5$ to $20$, the
constant can increase by $0.001$, the magnitude of the second
coefficients can increase by $0.03$, that of the third coefficient can
decrease by $0.3$.  We cannot conclude convincingly whether 
$F \rightarrow 1$ 
as $N \rightarrow \infty$.  




The corresponding fits for the embezzling family 
$\ket{\mu(f_{dh},n)}$ for $N_0 = 10$ are: 
\begin{eqnarray}
F_{+} & = 0.999982 - 0.377165/N + 0.282380/N^2
\nonumber
\\
F_{*} & = 0.999970 - 0.484107 / N + 0.359519 / N^2
\nonumber
\\
F_{\circ} & = 0.999960 - 0.565744 / N + 0.418400 / N^2
\nonumber
\end{eqnarray}
When $N_0$ ranges from $5$ to $20$, the constant can increase by
$0.0001$, the magnitudes of the second and third coefficients can
increase by $0.01$ and $0.1$.  

From the various fits, $\ket{\mu(f_{dh},n)}$ starts to outperform 
$\ket{\mu(n)}$ when $N \approx 140-160$.  

We note on the side that \cite{DH02} provides lower and upper bounds
on the optimal fidelity of embezzlement using $\ket{\mu(f_{dh}, n)}$.
We present the actual optimal performance (numerically) for small
$N$ that may be of interest elsewhere.  



\section{Discussions}

\vspace*{-2ex}


We have provided necessary conditions and sufficient conditions for
universal embezzling in the bipartite setting.  We exhibit an infinite
number of inequivalent families, present a family that outperforms
that proposed in \cite{DH02} for small $N$, but the latter {\em
  appears} optimal asymptotically based on our numerics.  
Our work does not resolve whether there is a regular or general 
universal embezzling
family achieving fidelity $1-O(1/(\log_2 n)^2)$.  We hope our
results are a step towards answering some of these questions.

\section{Acknowledgements}

\vspace*{-2ex}

We thank Wim van Dam, Aram Harrow, Patrick Hayden, Thomas Vidick, and
Volkher Scholz for their generous sharing of research results
concerning embezzlement, and Thomas Vidick for very useful comments on
an earlier version of the manuscript.  The fidelity of embezzling
$\ket{\varphi'}$ from (\ref{eq:aram-peter}) was calculated by Michal
Kotowski.



\vspace*{-1ex}

\newcommand{\etalchar}[1]{$^{#1}$}

\end{document}